\documentclass[11pt]{article}
\usepackage[total={6in, 8.1in}]{geometry}

\usepackage{epsfig}
\usepackage{cite}
\usepackage{amsmath}
\usepackage{amssymb}
\usepackage{color}
\usepackage[colorlinks,citecolor=red]{hyperref}
\usepackage{amsmath,amssymb,latexsym,amsthm,mathpazo,stmaryrd,dk}
\usepackage{eucal,enumerate}
\usepackage{tikz}
\usetikzlibrary{arrows,automata}
\usepackage{enumerate,xspace}
\usepackage[justification=centering, labelfont=bf, textfont=sl]{caption}

\newcommand\mcup\uplus
\newcommand\cat[1]{\mathsf{#1}}

\newtheorem{theorem}{Theorem}
\newtheorem{lemma}[theorem]{Lemma}

\begin{document}

\title{Abstract Huffman Coding and\\PIFO Tree Embeddings}

\author{%
Keri D'Angelo$^{\ast}$\hspace{1cm}Dexter Kozen$^{\dag}$\\[0.5em]
{\small\begin{minipage}{\linewidth}\begin{center}
\begin{tabular}{ccc}
Cornell University\\
Computer Science Department\\
Ithaca, New York 14853-7501, USA\\[2pt]
$^{\ast}$\href{mailto:kd349@cornell.edu}{\texttt{kd349@cornell.edu}}\hspace{1cm}$^{\dag}$\href{mailto:kozen@cs.cornell.edu}{\texttt{kozen@cs.cornell.edu}}
\end{tabular}
\end{center}\end{minipage}}
}

\maketitle

\begin{abstract}
Algorithms for deriving Huffman codes and the recently developed algorithm for compiling PIFO trees to trees of fixed shape \cite{MLFKK22a} are similar, but work with different underlying algebraic operations. In this paper, we exploit the monadic structure of prefix codes to create a generalized Huffman algorithm that has these two applications as special cases. 
\end{abstract}

\section{Introduction}

\emph{Huffman codes} translate letters from a fixed alphabet to $d$-ary codewords, achieving optimal compression for a given frequency distribution of letters. There is a well-known greedy algorithm for producing Huffman codes from a given distribution (see \cite{CoverThomas06}).

A new data structure called a \emph{PIFO tree} (priority-in first-out) has recently been proposed for implementing a wide range of packet scheduling algorithms in programmable network routers \cite{Sivaraman16,Alcoz20}. A PIFO tree is a tree of priority queues. Currently, most routers support just a few scheduling algorithms such as strict priority or weighted fair queueing, which are baked into the hardware. The schedulers can be configured to some extent, but it is generally not possible to implement more sophisticated scheduling algorithms that require reordering of already queued packets. This is exactly what PIFO trees permit. It seems likely that PIFOs will be supported on network devices in the near future.

Some researchers have already begun to explore how the PIFO abstraction can be emulated on conventional routers \cite{Alcoz20}. In very recent work \cite{MLFKK22a}, it was shown how to translate an algorithm designed for a PIFO tree of arbitrary shape to one that uses a PIFO tree of fixed shape, perhaps a complete $d$-ary tree that might be implemented in hardware, with negligible performance degradation.

The embedding algorithm is greedy and very similar to the Huffman algorithm, except that it is based on different algebraic operations. For Huffman coding, one wishes to choose a $d$-ary prefix code $C$ so as to minimize the value of $\sum_{x\in C} \len x\cdot r(x)$, where $r(x)$ is the frequency of the letter assigned to the codeword $x$. This minimizes the entropy of the resulting code. For PIFO trees, one wishes to minimize $\max_{x\in C} \len x + r(x)$, where $r(x)$ is the height of a subtree. This minimizes the height of the resulting $d$-ary tree and determines whether an embedding is at all possible. 

This similarity leads us to seek a unified axiomatic treatment that is parametric in the algebraic operations and that can be instantiated to produce both applications as special cases. Our treatment exploits the monadic structure of prefix codes to obtain an abstract formulation of the problem and its solution. We identify sufficient conditions for our abstract algorithm to produce optimal solutions, where the meaning of \emph{optimal} is also parametric in the instantiation.

We state axioms that are sufficient for optimality in \S\ref{sec:axioms}. The algorithm is presented in \S\ref{sec:algorithm} and its correctness proved in \S\ref{sec:correctness}. The two applications of Huffman codes and PIFO trees are derived in \S\ref{sec:applications}. 

\section{Background}
\label{sec:background}

We assume familiarity with the basic category-theoretic concepts of category, functor, and natural transformation. Our exposition is based on the concepts of \emph{monad} and \emph{Eilenberg-Moore algebra}; we briefly review the definitions here. For a more thorough introduction, we refer the reader to \cite{AspertiLongo91,BarrWells02,BarrWells90,Adamek09}.

\emph{Monads} are heavily used in functional programming to model the augmentation of a computation with extra structure \cite{Moggi91,Wadler92,Wadler95}. Formally, a \emph{monad} on a category $\cat C$ is a triple $(\TT, \eta, \mu)$, where $\TT : \cat C \to \cat C$ is an endofunctor on $\cat C$ and $\eta : I\to\TT$ and $\mu:\TT^2\to\TT$ are natural transformations, called the \emph{unit} and \emph{multiplication} respectively, such that for all objects $X$, the following diagrams commute:
\begin{align*}
&
\begin{tikzpicture}[->, >=stealth', auto, node distance=28mm]
\small
\node (NW) {$\TT^3 X$};
\node (NE) [right of=NW]{$\TT^2 X$};
\node (SW) [below of=NW, node distance=16mm]{$\TT^2 X$};
\node (SE) [right of=SW]{$\TT X$};
\path (NW) edge node {$\mu_{\TT X}$} (NE);
\path (NW) edge node[swap] {$\TT\mu_X$} (SW);
\path (SW) edge node[swap] {$\mu_X$} (SE);
\path (NE) edge node {$\mu_X$} (SE);
\end{tikzpicture}
&&
\begin{tikzpicture}[->, >=stealth', auto, node distance=28mm]
\small
\node (NW) {$\TT X$};
\node (NE) [right of=NW]{$\TT^2 X$};
\node (SW) [below of=NW, node distance=16mm]{$\TT^2 X$};
\node (SE) [right of=SW]{$\TT X$};
\path (NW) edge node {$\eta_{\TT X}$} (NE);
\path (NW) edge node[swap] {$\TT\eta_X$} (SW);
\path (SW) edge node[swap] {$\mu_X$} (SE);
\path (NE) edge node {$\mu_X$} (SE);
\path (NW) edge node[yshift=-3pt] {$\id_{\TT X}$} (SE);
\end{tikzpicture}
\end{align*}
Typical examples of monads are
\begin{itemize}
\item
the \emph{list monad}, in which $\eta_X(a) = \bx a$, the singleton list containing $a$, and
\begin{align*}
\mu_X(\bx{\bx{a_{11},\ldots,a_{1k_1}},\ldots,\bx{a_{n1},\ldots,a_{nk_n}}}) &= \bx{a_{11},\ldots,a_{1k_1},\ldots,a_{n1},\ldots,a_{nk_n}},
\end{align*}
the list flattening operation;
\item
the \emph{powerset monad}, in which $\eta_X(a) = \{a\}$, the singleton set containing $a$, and
$\mu_X(\AA) = \bigcup\AA$, the operation that takes a set of subsets of $X$ to its union.
\end{itemize}

Given a monad $(\TT,\eta,\mu)$ on a category $\cat C$, an \emph{Eilenberg-Moore algebra} for $(\TT,\eta,\mu)$ is a pair $(X,\gamma)$, where $X$ is an object of $\cat C$ and $\gamma:\TT X\to X$ is a morphism of $\cat C$, called the \emph{structure map} of the algebra, such that the following diagrams commute:
\begin{align*}
&
\begin{tikzpicture}[->, >=stealth', auto, node distance=28mm]
\small
\node (NW) {$\TT^2 X$};
\node (NE) [right of=NW]{$\TT X$};
\node (SW) [below of=NW, node distance=16mm]{$\TT X$};
\node (SE) [right of=SW]{$X$};
\path (NW) edge node {$\TT\gamma$} (NE);
\path (NW) edge node[swap] {$\mu_X$} (SW);
\path (SW) edge node[swap] {$\gamma$} (SE);
\path (NE) edge node {$\gamma$} (SE);
\end{tikzpicture}
&&
\begin{tikzpicture}[->, >=stealth', auto, node distance=28mm]
\small
\node (NW) {$X$};
\node (SW) [below of=NW, node distance=16mm]{$\TT X$};
\node (SE) [right of=SW]{$X$};
\path (NW) edge node[swap] {$\eta_X$} (SW);
\path (SW) edge node[swap] {$\gamma$} (SE);
\path (NW) edge node[yshift=-3pt] {$\id_{X}$} (SE);
\end{tikzpicture}
\end{align*}
A \emph{morphism of Eilenberg-Moore algebras} is a morphism of $\cat C$ that commutes with the structure maps. That is, if $(X,\gamma)$ and $(Y,\delta)$ are two algebras and $h:X\to Y$ is a morphism of $\cat C$, then $h$ is a morphism of algebras $h:(X,\gamma)\to(Y,\delta)$ if the following diagram commutes:
\begin{align*}
\begin{tikzpicture}[->, >=stealth', auto, node distance=28mm]
\small
\node (NW) {$\TT X$};
\node (NE) [right of=NW]{$\TT Y$};
\node (SW) [below of=NW, node distance=16mm]{$X$};
\node (SE) [right of=SW]{$Y$};
\path (NW) edge node {$\TT h$} (NE);
\path (NW) edge node[swap] {$\gamma$} (SW);
\path (SW) edge node[swap] {$h$} (SE);
\path (NE) edge node {$\delta$} (SE);
\end{tikzpicture}
\end{align*}
The Eilenberg-Moore algebras for $(\TT,\eta,\mu)$ and their morphisms form the \emph{Eilenberg-Moore} category over the monad $\TT$. The Eilenberg-Moore category for the list monad is the category of monoids and monoid homomorphisms. The Eilenberg-Moore category for the powerset monad is the category of complete upper semilattices and semilattice homomorphisms.

In our application, we will focus on the monad of \emph{$d$-ary prefix codes} on the category $\Set$ of sets and set functions. 

\section{Axioms}
\label{sec:axioms}

In this section, we state the axioms that are sufficient for the optimality of our generalized Huffman algorithm. 

Recall that a \emph{prefix code} over a fixed $d$-ary alphabet $\Sigma$ is a set of finite-length words over $\Sigma$ whose elements are pairwise incomparable with respect to the prefix relation. A prefix code $C$ is \emph{exhaustive} if every infinite $d$-ary string has a prefix in $C$. As a consequence of K\"onig's lemma, every exhaustive prefix code over a finite alphabet is finite, but not every finite prefix code is exhaustive.

Let $\CC:\Set\to\Set$ be an endofunctor in which
\begin{itemize}
\item
$\CC X$ is the set of pairs $(C,r)$ such that $C$ is a prefix code over a $d$-ary alphabet for some arbitrary but fixed $d\ge 2$ and $r:C\to X$, and
\item
for $h:X\to Y$, $\CC h:\CC X\to\CC Y$ with $\CC h(C,r)=(C,h\circ r)$.
\end{itemize}
The functor $\CC$ carries a natural monad structure with unit $\eta:I\to\CC$ and multiplication $\mu:\CC^2\to\CC$ defined by: for $a\in X$ and $(C,r)\in\CC^2X$ with $r(x)=(C_x,r_x)$,
\begin{align*}
\eta_X(a) &= (\{\eps\},\eps\mapsto a) &
\mu_X(C,r) &= (\set{xy}{x\in C,\ y\in C_x},xy\mapsto r_x(y)).
\end{align*}
The map $xy\mapsto r_x(y)$ is well defined, as the string $xy$ can be uniquely split into $x\in C$ and $y\in C_x$ because $C$ is a prefix code.

For example, consider the prefix codes $C = \{0,10,110,111\}$ and $C_0 = C_{10} = C_{110} = C_{111} = \{00,11\}$ over the binary alphabet $\{0,1\}$. The code $C$ is exhaustive but the others are not. Let
\begin{align*}
& r_0(00) = 2 && r_{10}(00) = 4 && r_{110}(00) = 6 && r_{111}(00) = 8\\
& r_0(11) = 3 && r_{10}(11) = 5 && r_{110}(11) = 7 && r_{111}(11) = 9\\
& r(0) = (C_0,r_0) && r(10) = (C_{10},r_{10}) && r(110) = (C_{110},r_{110}) && r(111) = (C_{111},r_{111}).
\end{align*}
Then $(C_0,r_0),(C_{10},r_{10}),(C_{110},r_{110}),(C_{111},r_{111})\in \CC\naturals$ and $(C,r)\in\CC^2\naturals$, and $\mu_\naturals(C,r) = (C',r')\in\CC\naturals$, where
\begin{align*}
& C' = \{000,011,1000,1011,11000,11011,11100,11111\}\\
& r'(000) = 2,\ r'(011) = 3,\ r'(1000) = 4,\ r'(1011) = 5,\\
& r'(11000) = 6,\ r'(11011) = 7,\ r'(11100) = 8,\ r'(11111) = 9.
\end{align*}

Suppose there is a fixed Eilenberg-Moore algebra $(W,w)$ with $w:\CC W\to W$. We call the elements of $W$ \emph{weights} and $(W,w)$ a \emph{weighting}. If $(C,r)\in\CC W$, then thinking of the elements of $C$ as a tree, the map $r:C\to W$ assigns a weight to each leaf of the tree, and the map $w$ tells how to assign a weight to the object $(C,r)$ based on the leaf weights $r$.

To define a notion of optimality, we assume that $W$ is totally preordered by $\le$; that is, $\le$ is reflexive and transitive, and for all $x,y\in W$, either $x\le y$ or $y\le x$ (or both). Smaller values of $W$ in the order $\le$ are considered better. We write $x\equiv y$ if both $x\le y$ and $y\le x$. Suppose further that we have a preorder on $\CC W$, also denoted $\le$, satisfying the following properties.

\begin{enumerate}[(i)]
\item
If $f:C\to D$ is bijective and length-nondecreasing, and if $r\le s\circ f$ pointwise, then $(C,r) \le (D,s)$. This says that longer codewords or larger leaf values cannot cause a decrease in the order $\le$.
\item
(Exchange property) If $r(x)\le r(y)$, $\len x\le \len y$, and
\begin{align*}
s(z) &= \begin{cases}
r(x), & \text{if $z=y$},\\
r(y), & \text{if $z=x$},\\
r(z), & \text{if $z\in C\setminus\{x,y\}$},
\end{cases}
\end{align*}
then $(C,s)\le (C,r)$. That is, it never hurts to swap a larger element deeper in the tree with a smaller element higher in the tree.
\item
The monad structure maps $\eta_W:W\to\CC W$ and $\mu_W:\CC^2 W\to\CC W$ are monotone with respect to $\le$, where $\le$ on $\CC^2 W$ is defined by:
\begin{align*}
(C,r)\le(D,s)\ \Iff\ \CC w(C,r)\le\CC w(D,s).
\end{align*}
\end{enumerate}

Some special cases of (i) are
\begin{itemize}
\setcounter{enumi}3
\item
If $f:C\to D$ is bijective and length-nondecreasing, then $(C,s\circ f)\le(D,s)$. Thus lengthening codewords cannot cause $\le$ to decrease.
\item
If $f:C\to D$ is bijective and length-preserving, then $(C,s\circ f)\equiv(D,s)$. This says that the order $\le$ on trees depends only on the lengths of the codewords in $C$, not on the actual codewords themselves.
\item
If $r,s:C\to W$ and $r\le s$ pointwise, then $(C,r)\le (C,s)$. Thus larger leaf values cannot cause $\le$ to decrease.
\end{itemize}
We assume these properties hold for the algorithm described in the next section. 

For $(C,r),(D,s)\in\CC W$, let us write $(C,r)\sim(D,s)$ if the multisets of weights represented by the two objects are the same; that is, there is a bijective function $f:C\to D$ such that $r=s\circ f$. A tree $(C,r)\in\CC W$ is defined to be \emph{optimal} (for its multiset of weights) if $(C,r)$ is $\le$-minimum in its $\sim$-class; that is, $(C,r)\le(D,s)$ for all $(D,s)$ such that $(C,r)\sim(D,s)$.

We will give two detailed examples in \S\ref{sec:applications}.

\section{Algorithm}
\label{sec:algorithm}

Suppose we are given a multiset $M$ of weights in $W$, $\len M\ge 2$. We would like to find an optimal tree for this multiset of weights. The following is a recursive algorithm to find such an optimal tree.

\begin{enumerate}
\item
Say there are $n\ge 2$ elements in $M$.
Let $k\in\{2,\ldots,d\}$ such that $n\equiv k \bmod (d-1)$. Let $\seq a0{k-1}$ be the $k$ elements of least weight. Form the object
\begin{align*}
(\{0,1,\ldots,k-1\},i\mapsto a_i)\in\CC W.
\end{align*}
If there are no other elements of $M$, return that object.
\item
Otherwise, let
\begin{align*}
M' &= \{(\{0,1,\ldots,k-1\},i\mapsto a_i)\}\cup\set{\eta_W(a)}{a\in M\setminus\{\seq a0{k-1}\}},
\end{align*}
a multiset of $n-k+1<n$ elements of $\CC W$.
\item
Recursively call the algorithm at step 1 with $M'' = \set{w(E,t)}{(E,t)\in M'}$, a multiset of elements of $W$. This returns a tree $(D,s)$ of type $\CC W$ that is optimal for $M''$. The bijective map $s:D\to M''$ factors as $w\circ s'$ for some bijective $s':D\to M'$, and $(D,s')\in\CC^2 W$ with $\CC w(D,s') = (D,w\circ s') = (D,s)$.
Flatten this to $\mu_W(D,s')\in\CC W$ and return that value.
\end{enumerate}

Note that the number of items combined in step 1 will be $d$ in all recursive calls except possibly the first. This is because in every step, if $k\in\{2,3,\ldots,d\}$, then after that step the number of remaining elements will be $(c(d-1) + k) - k + 1 = c(d-1) + 1$, which is congruent to $d$ mod $d-1$, so $d$ elements will be taken in the next step. But from that point on, it is an invariant of the recursion that the number of elements remaining is $1$ mod $d-1$, since in each step we remove $d$ elements and add one back, decreasing the number by $d-1$.

\section{Correctness}
\label{sec:correctness}

In this section, we prove the correctness of the algorithm, making use of the following lemma.

\begin{lemma}
\label{lem:km}
Let $k\in\{2,3,\ldots,d\}$ and $k\equiv\len M\bmod (d-1)$. Let $\seq a0{k-1}$ be the $k$ elements of $M$ of least weight, listed in nondecreasing order of weight. There is an optimal tree in $\CC W$ in which $\seq a0{k-1}$ are sibling leaves at the deepest level and have no other siblings.
\end{lemma}
\begin{proof}
Let $(C,r)\in\CC W$ be optimal. Axiom (i) allows us to transform $(C,r)$ so that there are no deficient nodes (nodes with fewer than $d$ children) at any level except the deepest, and only one deficient node at the deepest level. Thus we can assume without loss of generality that there are $k$ elements $\seq x0{k-1}\in C$ of maximum length $n$ in $C$ with a common prefix of length $n-1$, and no other $y\in C$ has that prefix. Say the $\seq x0{k-1}$ are listed in nondecreasing order of $r(x_i)$; that is, $r(x_i)\le r(x_j)$ for all $0\le i\le j\le k-1$. Let $\seq y0{k-1}\in C$ such that $r(y_i)=a_i$. Since the $a_i$ are minimal, $r(y_i)\le r(x_i)$. Because the $\len{x_i}$ are of maximum length, $\len{y_i} \le \len{x_i}$. Now we can swap using axiom (ii). Let 
\begin{align*}
s(z) &= \begin{cases}
r(x_i), & \text{if $z=y_i$},\\
r(y_i), & \text{if $z=x_i$},\\
r(z), & \text{otherwise}.
\end{cases}
\end{align*}
Then $(C,s) \le (C,r)$. But since $(C,r)$ was optimal, $(C,r)\equiv(C,s)$ and $(C,s)$ is also optimal.
\end{proof}

\begin{theorem}
The algorithm of \S\ref{sec:algorithm} produces an optimal tree.
\end{theorem}
\begin{proof}
By induction on $n$. The basis is $n\le d$, in which case the result is straightforward.

Suppose that we have a multiset $M$ of $n>d$ elements of $W$. Let $(C,r)$ be an optimal tree for $M$. Let $k\in\{2,3,\ldots,d\}$ be congruent mod $d-1$ to $\len{M}$. Let $\seq a0{k-1}$ be the $k$ smallest elements of $M$. By Lemma \ref{lem:km}, we can assume without loss of generality that $\seq a0{k-1}$ are siblings and occur at maximum depth in $(C,r)$, so there exist strings $x\kern1pt0,x\kern1pt1,\ldots,x\kern1pt(k-1)\in C$ of maximum length with a common prefix $x$ and $r(x\kern1pti)=a_i$. Remove the strings $x\kern1pti$ from $C$ and replace them with $x$. Call the resulting set $C'$. For $z\in C'$, let
\begin{align*}
r'(z) &= \begin{cases}
(\{0,1,\ldots,k-1\},i\mapsto a_i), & \text{if $z=x$},\\
\eta_W(r(z)), & \text{otherwise}.
\end{cases}
\end{align*}
Then $(C',r')\in\CC^2 W$ and $(C,r)=\mu_W(C',r')$. The multiset of values of $r'$ is just the $M'$ of step 2 of the algorithm.

The algorithm will form the multiset
\begin{align*}
M'' = \set{w(E,t)}{(E,t)\in M'} = \set{w(r'(z))}{z\in C'}
\end{align*}
and recursively call with these weights. By the induction hypothesis, the return value will be a tree $(D,s)\in\CC W$ that is optimal for $M''$, thus $(D,s)\le(C',w\circ r')$, and the bijective map $s:D\to M''$ factors as $s=w\circ r'\circ f$ for some bijective $f:D\to C'$. Let $s'=r'\circ f$. By axiom (iii),
\begin{align*}
\CC w(D,s') = (D,w\circ s') = (D,s) \le (C',w\circ r') = \CC w(C',r'), 
\end{align*}
therefore
$(D,s')\le(C',r')$, and since $\mu_W$ is monotone,
\begin{align*}
\mu_W(D,s')\le\mu_W(C',r') = (C,r).
\end{align*}
As $(C,r)$ was optimal, so is $\mu_W(D,s')$, and this is the value returned by the algorithm. 
\end{proof}

\section{Applications}
\label{sec:applications}

By choosing two specific weightings $(W,w)$ and defining the ordering relations $\le$ appropriately, we can recover two special cases of this algorithm.

\subsection{Huffman coding}

Our first application is Huffman codes. Here we wish to minimize the expected length of variable-length codewords, given frequencies of the letters to be coded. For this application, we take $W=\reals_+=\set{a\in\reals}{a\ge 0}$ with weighting
\begin{align*}
w(C,r) &= \sum_{x\in C}r(x).
\end{align*}
Recall that for $a\in W$ and $(C,r)\in\CC^2 W$ with $r(x)=(C_x,r_x)$,
\begin{align*}
\eta_W(a) &= (\{\eps\},\eps\mapsto a) &
\mu_W(C,r) &= (\set{xy}{x\in C,\ y\in C_x}, xy\mapsto r_x(y)).
\end{align*}
Then $(W,w)$ is an Eilenberg-Moore algebra for the monad $(\CC,\mu,\eta)$, as
\begin{align*}
w(\eta_W(a)) &= w(\{\eps\},\eps\mapsto a)
= \sum_{x\in\{\eps\}}(\eps\mapsto a)(x)
= a,\\
w(\mu_W(C,r)) &= \sum_{x\in C}\sum_{y\in C_x}r_x(y)
= \sum_{x\in C}w(C_x,r_x)\\
&= \sum_{x\in C}w(r(x))
= w(C,w\circ r)
= w(\CC w(C,r)).
\qedhere
\end{align*}

In addition, let us define $\alpha:\CC W\to W$ by
\begin{align*}
\alpha(C,r) &= \sum_{x\in C}\len x\cdot r(x).
\end{align*}

\begin{lemma}
\label{lem:Huffman}
\begin{align*}
\alpha(\eta_W(a)) &= 0 &
\alpha(\mu_W(C,r)) &= \alpha(C,w\circ r) + w(C,\alpha\circ r).
\end{align*}
\end{lemma}
\begin{proof}
\begin{align*}
\alpha(\eta_W(a)) &= \alpha(\{\eps\},\eps\mapsto a)
= \sum_{x\in\{\eps\}}\len x\cdot(\eps\mapsto a)(x)
= \len\eps\cdot a
= 0,\\
\alpha(\mu_W(C,r))
&= \alpha(\set{xy}{x\in C,\ y\in C_x}, xy\mapsto r_x(y))\\
&= \sum_{x\in C}\sum_{y\in C_x}\len{xy}\cdot r_x(y)
= \sum_{x\in C}\len{x}\sum_{y\in C_x}r_x(y) + \sum_{x\in C}\sum_{y\in C_x}\len{y}\cdot r_x(y)\\
&= \sum_{x\in C}\len{x}\cdot w(C_x,r_x) + \sum_{x\in C}\alpha(C_x,r_x)
= \sum_{x\in C}\len{x}\cdot w(r(x)) + \sum_{x\in C}\alpha(r(x))\\
&= \alpha(C,w\circ r) + w(C,\alpha\circ r).
\qedhere
\end{align*}
\end{proof}

Note that $\alpha$ and $w$ agree on trees of depth one:
\begin{align*}
w(\{0,1,\ldots,k-1\},i\mapsto a_i) &= \sum_{i=0}^{k-1} a_i,\\
\alpha(\{0,1,\ldots,k-1\},i\mapsto a_i) &= \sum_{i=0}^{k-1} \len i\cdot a_i = \sum_{i=0}^{k-1} a_i,
\end{align*}
where $\len i$ refers to the length of $i$ as a string, which in this case is 1.

The map $\alpha$ is related to the Shannon entropy $H$. If $r(x) = d^{-\len x}$, the probability of a $d$-ary codeword $x$ under the uniform distribution on a $d$-ary alphabet, then
\begin{align*}
H(C,r) &= \sum_{x\in C} -d^{-\len x}\log d^{-\len x}
= \sum_{x\in C} \len x\cdot d^{-\len x}\log d
= \alpha(C,r)\log d,
\end{align*}
so $\alpha(C,r) = H(C,r)/\log d$.

To use the algorithm in \S\ref{sec:algorithm}, we need an order $\le$ on $\CC W$. Define $(C,r)\le (D,s)$ if $(C,r)\sim(D,s)$, that is, there is a bijective map $f:C\to D$ such that $r=s\circ f$, and
\begin{align*}
\alpha(C,r)\le\alpha(D,s).
\end{align*}
Note that if $(C,r)\le (D,s)$, then
\begin{align*}
w(C,r) &= \sum_{x\in C}r(x) = \sum_{x\in C}s(f(x)) = \sum_{y\in D}s(y) = w(D,s).
\end{align*}
According to axiom (iii), for $(C,r),(D,s)\in\CC^2 W$,
\begin{align}
(C,r)\le(D,s)\ &\Iff\ \CC w(C,r)\le\CC w(D,s)\nonumber\\
&\Iff\ \alpha(\CC w(C,r))\le\alpha(\CC w(D,s))\nonumber\\
&\Iff\ \alpha(C,w\circ r)\le\alpha(D,w\circ s).\label{eq:py2}
\end{align}
Also, if $(C,r)\le (D,s)$ in $\CC^2 W$, then
\begin{align}
w(C,\alpha\circ r) &= \sum_{x\in C}\alpha(r(x)) = \sum_{x\in C}\alpha(s(f(x))) = \sum_{y\in D}\alpha(s(y)) = w(D,\alpha\circ s).\label{eq:py3}
\end{align}

\begin{lemma}
$\mu_W:\CC^2 W\to\CC W$ and $\eta_W:W\to\CC W$ are
monotone with respect to $\le$.
\end{lemma}
\begin{proof}
For $\eta_W$, suppose $a,b\in W$ and $a\le b$. By Lemma \ref{lem:Huffman},
\begin{align*}
\alpha(\eta_W(a)) &= 0 = \alpha(\eta_W(b)) &
w(\eta_W(a)) &= a \le b = w(\eta_W(b)).
\end{align*}
For $\mu_W$, suppose $(C,r),(D,s)\in\CC^2 W$ and $(C,r)\le(D,s)$. By Lemma \ref{lem:Huffman}, \eqref{eq:py2}, and \eqref{eq:py3},
\begin{align*}
\alpha(\mu_W(C,r))
&= \alpha(C,w\circ r) + w(C,\alpha\circ r)\\
&\le \alpha(D,w\circ s) + w(D,\alpha\circ s)
= \alpha(\mu_W(D,s)).
\qedhere
\end{align*}
\end{proof}

\begin{theorem}
The algorithm in \S\ref{sec:algorithm} for the algebra $(\reals_+, w)$ and ordering relation $\le$ defined by $\alpha$ is equivalent to Huffman's algorithm and produces an optimal Huffman code for a given multiset of weights.
\end{theorem}
\begin{proof}
Take $X\subset \reals_+$ to be a finite multiset and sort the set $X$ in increasing order. For the binary case of Huffman codes (the $d$-ary version follows the same way), we always choose $k=2$. For the first step, let $a_0, a_1 \in X$ be the two smallest elements in the list. Form the object $(\{ 0, 1 \}, i \mapsto a_i) \in \CC X$. In the case $n=2$, this is the only remaining object in the list. Otherwise, we combined them into one element with the sum of the weights of $a_0$ and $a_1$ as the weight of the new element, exactly as the Huffman coding does.

For the case $n>2$, there are remaining elements in the set $X$. Take all remaining $a \in X \backslash \{a_0, a_1\}$ and replace $a$ by $\eta_X(a) \in \CC X$. We are left with $n-1$ elements of type $\CC X$. If we recursively call the algorithm in step 1, we are continually combining the least two elements in the remaining set with the elements weighted by $w$. Note by the weighting $w$, $w(\eta_X(a)) = a$ and on elements in $\CC X$, $w$ takes the sum of $r(x)'s$, exactly as Huffman coding does. Finally, this leaves us with a tree in $\CC^2 X$ where leaves have weights of the form $\eta_X(a_i)$. Denote this tree by $(D,s)$. Taking $\mu_X(D,S)$ gives our desired tree in $\CC X$. 
\end{proof}

\subsection{PIFO trees}

PIFO trees were introduced in \cite{Sivaraman16} as a model for programmable packet schedulers. In the recent work of \cite{MLFKK22a}, further work was done on PIFO trees giving a semantics that allows for certain embedding algorithms. The notion of a \emph{homomorphic embedding} was defined for the purpose determining when a PIFO tree could be represented by another PIFO tree and for finding an embedding if so. The embedding algorithm we consider takes an arbitrary PIFO tree and embeds it into a $d$-ary tree. This becomes a special case of the algorithm of \S\ref{sec:algorithm}, where we choose $w$ in the weighting $(W,w)$ to minimize the height of the target $d$-ary tree into which the source tree can embed.

For this application, we take $W=\naturals$ with weighting
\begin{align*}
w(C,r) &= \max_{x\in C}\len x + r(x).
\end{align*}
This gives an Eilenberg-Moore algebra $(W,w)$ for the monad $(\CC,\mu,\eta)$.
For $a\in W$ and $(C,r)\in\CC^2 W$ with $r(x)=(C_x,r_x)$, as before we have
\begin{align*}
\eta_W(a) &= (\{\eps\},\eps\mapsto a) &
\mu_W(C,r) &= (\set{xy}{x\in C,\ y\in C_x}, xy\mapsto r_x(y)),
\end{align*}
so
\begin{align*}
w(\eta_W(a)) &= w(\{\eps\},\eps\mapsto a)
= \max_{x\in\{\eps\}}\len x + (\eps\mapsto a)(x)
= \len\eps + a
= a,\\
w(\mu_W(C,r))
&= w(\set{xy}{x\in C,\ y\in C_x}, xy\mapsto r_x(y))
= \max_{x\in C}\max_{y\in C_x}\len{xy} + r_x(y)\\
&= \max_{x\in C}\max_{y\in C_x}\len{x} + \len{y} + r_x(y)
= \max_{x\in C}\len{x} + \max_{y\in C_x}\len{y} + r_x(y)\\
&= \max_{x\in C}\len{x} + w(C_x,r_x)
= \max_{x\in C}\len{x} + w(r(x))\\
&= w(C,w\circ r)
= w(\CC w(C,r)).
\end{align*}

For $(C,r),(D,s)\in\CC W$, let us define $(C,r)\le (D,s)$ if there is a bijective function $f:C\to D$ such that $r=s\circ f$ and
\begin{align*}
w(C,r)\le w(D,s).
\end{align*}
\begin{lemma}
$\mu_W:\CC^2 W\to\CC W$
and $\eta_W:W\to\CC W$ are
monotone with respect to $\le$.
\end{lemma}
\begin{proof}
For $\eta_W$, if $a\le b$, then $w(\eta_W(a)) = a \le b = w(\eta_W(b))$.

For $\mu_W$, suppose $(C,r),(D,s)\in\CC^2 W$ and $(C,r)\le(D,s)$. According to axiom (iii),\begin{align*}
(C,r)\le(D,s)\ &\Iff\ \CC w(C,r)\le \CC w(D,s)\\
&\Iff\ w(\CC w(C,r))\le w(\CC w(D,s)).
\end{align*}
Then
\begin{align*}
w(\mu_W(C,r))
&= w(\CC w(C,r))
\le w(\CC w(D,s))
= w(\mu_W(D,s)).
\qedhere
\end{align*}
\end{proof}

\begin{theorem}
The algorithm of \S\ref{sec:algorithm} for the algebra $(\naturals, w)$ and ordering relation $\le$ defined by $w$ is equivalent to determining whether an embedding of a PIFO tree in a bounded $d$-ary tree exists and finding the embedding if so.
\end{theorem}

\section{Conclusion}

We have presented a generalized Huffman algorithm and shown that two known algorithms, Huffman codes and embedding of PIFOs trees, can be derived as special cases. The PIFO embedding algorithm was introduced in \cite{MLFKK22a} and observed to be very similar to the usual combinatorial algorithm for optimal Huffman codes, albeit based on a different algebraic structure. This suggested the common generalization presented in this paper.

Our generalized algorithm exploits the monadic structure of prefix codes, which allows a more algebraic treatment of the Huffman algorithm than the usual combinatorial approaches. The two applications fit naturally in the categorical setting by choosing specific Eilenberg-Moore algebras for each one. It is possible that other greedy algorithms might fit into this framework as well.

%\bibliographystyle{IEEEbib}
%\bibliography{PIFO,Cat}

\end{document}